\documentclass[final,leqno]{siamltex}

\usepackage{latexsym}
\usepackage{epsfig,epstopdf}
\usepackage{amsmath}
\usepackage{amssymb}
\usepackage{color}
\usepackage{multirow}
\usepackage{hhline}
\usepackage[dvipdfm,CJKbookmarks,bookmarksopen=true,colorlinks=true,
linkcolor=blue,citecolor=blue,pdfstartview=FitH,pdftitle=title,pdfauthor=lixx]{hyperref}
\usepackage{amsfonts}
\usepackage{graphicx}
\usepackage{color}
\usepackage{algorithm}
\usepackage{algorithmic}
\usepackage{tikz}
\usepackage{verbatim}
\usepackage{bbm}
\usepackage{bm}
\usetikzlibrary{arrows,calc}

\newtheorem{thm}{Theorem}
\newtheorem{lem}[thm]{Lemma}

\newtheorem{assu}{\textit{Assumption}}

\newcommand{\beq}{\begin{equation}}
\newcommand{\eeq}{\end{equation}}
\newcommand{\bey}{\begin{eqnarray}}
\newcommand{\eey}{\end{eqnarray}}

\newcommand{\mmm}[1]{{\left\vert\kern-0.25ex\left\vert\kern-0.25ex\left\vert #1
    \right\vert\kern-0.25ex\right\vert\kern-0.25ex\right\vert}}
 
\newcommand{\mm}[1]{{\left\vert\kern-0.25ex\left\vert #1
    \right\vert\kern-0.25ex\right\vert}}

\newcommand{\ep}{\epsilon}

\newcommand{\xb}{\mathbf{x}}

\newcommand{\ub}{{\bf u}}

\newcommand{\yb}{{\bf y}}

\newcommand{\Kb}{{\bf K}}
\newcommand{\Fb}{{\bf F}}

\newcommand{\Od}{{\Omega}}

\newcommand{\omd}{\omega_\delta}
\newcommand{\rds}{\rho_\delta(s)}

\newcommand{\del}{\delta}

\newcommand\mbG{{\mathbb{G}_{\rho}^\delta}}
\newcommand\mGz{{\mathbb{G}_{0}}}
\newcommand\mfG{{\mathbb{G}_{\rho}^\delta}}

\begin{document}

\title{Stability of nonlocal Dirichlet integrals and implications for peridynamic correspondence material modeling\thanks{This 
work is supported in part by the U.S.~NSF grant
DMS-1719699, and AFOSR MURI center for material failure prediction through peridynamics.}}


\author{Qiang Du\thanks{Department of Applied Physics and Applied Mathematics, 
Columbia University, New York, NY 10027; {\tt qd2125@columbia.edu}}
\and  Xiaochuan Tian\thanks{Department of Applied Physics and Applied Mathematics, 
Columbia University, New York, NY 10027; Present address:
Department of Mathematics, University of Texas, Austin,
TX 78712; {\tt xtian@math.utexas.edu}}}


\maketitle

\begin{abstract}
Nonlocal gradient operators are basic elements of nonlocal vector calculus that play
 important roles in nonlocal modeling and analysis. 
 In this work, we extend earlier analysis on nonlocal gradient
  operators. In particular, we study a 
  nonlocal Dirichlet integral
 that is given by a quadratic energy functional 
  based on  nonlocal gradients. Our main finding, which differs from claims made in previous
  studies,  is that the coercivity and stability
  of this nonlocal continuum energy functional may hold for some properly chosen nonlocal interaction
  kernels but may fail for some other ones. This can be significant for possible applications 
  of nonlocal gradient operators in various nonlocal models. In particular, we discuss some
  important implications  for the  peridynamic correspondence material models. 
\end{abstract}



\pagestyle{myheadings} \thispagestyle{plain}
\markboth{DU AND TIAN}{STABILITY OF DIRICHLET INTEGRALS WITH NONLOCAL GRADIENT}


\section{Introduction}
\label{intro}
\setcounter{equation}{0}
Recently, there have been much interests developing nonlocal models for a variety of problems 
arising in physics, biology,  materials and social sciences \cite{heat,image1,image2,image3,Tad2}. 
A nonlocal model made up by nonlocal integral operators
can potentially allow more singular solutions than the classical differential equation counterpart, thus offering 
great promise in the effective modeling of singular defects and anomalous properties such as 
cracks and fractures \cite{silling1}.
Nonlocal gradient operators are basic elements of nonlocal vector calculus that play
 important roles in nonlocal modeling and analysis  \cite{sys1,sys2,MeDu15}. 
 The development of a systematic mathematical framework for nonlocal
problems, in parallel to that for local classical partial differential equations (PDEs), in turn has
provided foundation and clarity to practical nonlocal modeling techniques such as peridynamics.
In particular, 
the rigorous mathematical studies of nonlocal gradient operators
  have found successful  applications ranging from nonlocal gradient recovery for
  robust a posteriori stress analysis in nonlocal mechanics to nonlocal in time modeling
  of anomalous diffusion \cite{dtty16,dyz16}.  Asymptotically compatible schemes to discretize
  the nonlocal gradient operators have also been presented in \cite{dtty16}, following the
  framework given in \cite{non10,non11}.

In this work, we continue our analysis of nonlocal gradient operators initiated in \cite{sys1,sys2} and
further explored in \cite{dtty16,dyz16,MeDu15}. We address the coercivity and stability
of energy functionals with the energy density formed by the nonlocal gradient. As a representative example,
 the functional considered here is what we refer to as a nonlocal Dirichlet integral, which is simply  a 
 quadratic energy of the nonlocal gradient. Given that
the most popular forms of the nonlocal gradients are often exclusively determined by
some nonlocal
interaction kernels (micro-modulus functions),  the issue of coercivity and stability of simple
energies like the Dirichlet integral rests largely
on properties of the underlying interaction kernel  functions (beside possibly boundary conditions or nonlocal 
constraints). While there were claims made in the literature on the universal loss of coercivity and stability
for all feasible kernels,  
we show that, contrary to such claims, there is a 
 class of kernels that can assure the coercivity and stability.  The coercviity result established here is
 quite strong in the sense that it holds uniformly with respect to the horizon parameter that measures
 the range of nonlocality, so that,
 in the local limit, we recover the well-known coercivity result of the classical, local Dirichlet integrals
 formed by the conventional local derivative. Providing conditions to characterize
 this class of kernels is the main
contribution of this work.

Our main finding has a number of implications for nonlocal modeling. In particular, 
in the context of peridynamic correspondence material models (see a brief description in the next section
and  additional discussions in 
 \cite{silling1,flux} and \cite{bessa14,dtty16,ric,tr14}), we can conclude that 
there may not be any loss of stability if proper kernels are used for the correspondence formulation.
This is an encouraging news to the community that has found convenience in using
 peridynamic correspondence material models.  However, identifying the right kernels is crucial,
 in addition to requiring their consistency to the underlying physical processes and principles,  
 In simple terms,
the requirement on the kernels to assure the energy coercivity is that one should suitably enforce 
 stronger interactions  as the undeformed bond length
gets shorter. The specific form of the strengthening is given
 in the assumption and theorems in  section  \ref{sec:sta}, 
We also point out that, in contrast, the weakening of interactions   (or the lack of sufficient strengthening)
among materials points in closer proximity
 is likely going to cause a loss of coercivity of the correspondence formulation.
 In short, this means that the choice of nonlocal interaction kernels is a much more subtle issue for the
 correspondence theory, in 
 comparison with, for example, other nonlocal formulations such as the bond-based or state-based peridynamics of linear
 elasticity \cite{MeDu14e,MeDu14r}.
 
 To avoid technical complications, we only
 present the derivation in one space dimension for a scalar field, though the extension to multidimensional
 cases and vector fields are immediate, based on similar calculations  give in \cite{dt17}.  
 For simplicity, we also only consider periodic
 boundary conditions to avoid the discussion near physical boundary.  By utilizing this special geometry,
 we can carry out the needed mathematical derivations using simple Fourier analysis and elementary calculations. 
 The extensions to more general
 boundary conditions or more appropriate nonlocal volumetric constraints \cite{sys1,sys2} are more
 involved and will be left as future work. Further studies on the discrete level can also be
 carried out in a similar fashion but it is beyond the scope of the current work so that  in this work 
 we can focus on
 delivering a simple but importance message on nonlocal correspondence models on the continuum level.
 Indeed, as shown in \cite{non10,non11} and again in this work, delineating 
the effects on physical and continuum scales from those arising from numerical resolutions 
in order to better  investigate their interplay has proven to be a helpful strategy to validate nonlocal modeling
and simulations. Moreover, the conditions given later on the nonlocal interaction kernels rule out many popular choices
used in existing simulation codes and applications. Again, this is another instance related to nonlocal modeling
and simulations where
popular practices in the past may need to be
carefully scrutinized. As another example, simple mid-point quadrature and piecewise constant Galerkin finite element
approximations are both popular discretizations of bond-based peridynamics but they are not robust and run the
risk of converging to wrong physical solutions \cite{non10,non11}. Ensuring properly defined models and
convergent algorithms is particularly important to subjects like peridynamics since their main goal is to deal with
complex systems involving multiscale features, patterns and singular solutions potentially generated by inherent 
material instabilities so that one does not mix up model or numerical instability with the physically reality.

The remainder of the paper is organized as follows. In section \ref{sec:bgd}, some general background on the subject
is given.
We then provide more discussions on the nonlocal gradient operator for   one dimensional scalar field   in section \ref{sec:ngo}.
Equivalent formulations are presented in section \ref{sec:equ} to draw connections with nonlocal diffusion operators.
The main stability analysis is
in section \ref{sec:sta}. Different variants are considered in section \ref{sec:oth}.
Finally, 
in section \ref{sec:con}, we 
make some conclusions on the implications and generalizations.

\section{Background}
\label{sec:bgd}

In nonlocal mechanical models and nonlocal diffusion equations, the primary quantities used 
are often displacement and density variables. Drawing analogy to classical and local models, 
the notion of nonlocal gradient is indispensable as it is related naturally to concepts of nonlocal strain and
stress and nonlocal flux, see \cite{flux,silling1} and \cite{bessa14,dtty16,ric,tr14}. 

Indeed, it is a widely accepted practice of continuum mechanics to use classical local gradients to help defining
 constitutive relations that are central to the underlying mathematical models.  Taking the peridynamic model
in $\mathbb{R}^d$  as 
an example, if we use $\ub$ to denote the deformation vector field from
 $\mathbb{R}^d$ to $\mathbb{R}^d$, the
so-called nonlocal deformation gradient tensor $\Fb(\xb)=\mbG{\bf u}(\xb)$ is given by
\begin{equation}
\label{eq:ndg}
\mbG{\bf u}(\xb)=\left(\int_{ B_\delta(\xb)}
 \rho_\delta(|\yb -\xb   |)\,  \frac{\yb -\xb }{|\yb -\xb   |} \otimes
\, \frac{{\bf u}(\yb )-{\bf u}(\xb   )}{|\yb -\xb   |} \;d\yb\right) \mathbf{K}^{-1}(\xb) \;\;
\end{equation}
where $\delta>0$ represents the horizon measuring the range of nonlocal interactions,
  $\rho_\delta$ is a scalar function
   (micromodulus function as named in \cite{silling1}) representing the nonlocal interaction kernel, and $\Kb=\Kb(\xb)$ is
  the shape tensor defined by
  $$ 
  \mathbf{K}(\xb)=\int_{ B_\delta(\xb)}  \rho_\delta(|\yb -\xb   |)\,  \frac{\yb -\xb }{|\yb -\xb   |} \otimes
\, \frac{\yb -\xb }{|\yb -\xb   |} \;d\yb\, .
$$
It has been suggested that, corresponding to a well-defined local
constitutive model $\sigma=\sigma(\varepsilon)$ where $\sigma$ is the stress tensor and $\varepsilon$ the
strain tensor, one can formally derive a nonlocal peridynamic analog as   $\sigma=\sigma(\bar{\Fb})$ where
$\bar{\Fb}$ is the symmetric part of $\Fb$.
This in essence leads to the  peridynamic correspondence material models or correspondence theory for short \cite{silling2}.
More discussions and additional references in the context of peridynamics can be found in \cite{PD0,PD1,PD2,foster16,PD3,PD4,tr14}, among others.

We note that in more mathematical generality, a nonlocal gradient operator for vector field $\ub$ defined on
a domain $\Omega$ may take on the form of a second order tensor given by
$$
\mfG{\bf u}(\xb   ):= \lim_{\epsilon \to 0} \, \int_{\Omega\setminus B_\epsilon(\xb   )} \rho_\delta(|\yb -\xb   |)\,\mathsf{M}_\delta(\yb -\xb   ) \, \frac{{\bf u}(\yb )-{\bf u}(\xb   )}{|\yb -\xb   |} \;d\yb ,\;\;
$$
where  $\mathsf{M}_\delta$  is a 3rd-order odd  tensor, thanks to the Schwartz kernel theorem \cite{MeDu15}. The simpler version given in \eqref{eq:ndg} is nevertheless sufficient to serve the purpose of our discussion here.

Though the form given in \eqref{eq:ndg} is appealing and is formally  consistent to the classical deformation gradient
in the local limit with suitably normalized nonlocal interaction kernel $\rho_\delta$, 
there have been various issues concerning its use in the correspondence theory \cite{bessa16,silling2,tr14}. Similar issues 
have been noticed in other applications such as those involving particle discretizations \cite{bessa14}. Some of these issues
are related to numerical implementations but there are also fundamental limitations on the level
of the continuum models, see \cite{silling2}
for a recent study that provided a comprehensive summary on the topic. In particular, it has been made aware of the loss
of coercivity (stability) of the nonlocal energy functionals constructed explicitly  via the nonlocal deformation gradient.
We may use
Dirichlet integrals  as representative examples of energy functionals given by a quadratic energy density
corresponding to the linear elasticity. With the nonlocal deformation gradient $\Fb=\mbG\ub$, the nonlocal
Dirichlet integral is given by
\begin{equation}
\label{eq:dir}
E_\delta(\ub)=\int |\bar{\Fb}(\xb)|^2 d\xb=\int \left|\frac{\mbG\ub(\xb)+(\mbG\ub(\xb))^T}{2}
\right|^2 d\xb\,,
\end{equation}
which is an nonlocal analog of the usual Dirichlet integral defined for the local gradient as follows:
$$
E_0(\ub)=\int  \left|\frac{\nabla \ub (\xb)+(\nabla \ub (\xb))^T}{2}\right|^2 d\xb\, .
$$
The lack of coercivity in the nonlocal version is unfortunate as the the local version 
is well-known to be coercive subject to suitable boundary conditions or constraints.
Despite the existing issues and improved understanding, formulations based correspondence theory continue to
be popular among practitioners and, at the same time, they are also challenged by the community given the lingering debate
on the relevant controversies.
The latest attempt \cite{silling2} has suggested the addition  of a penalty term  that
 does provide the needed stability to the elastic energy on
the continuum level, however, the additional term to the elastic energy generically is not a null-Lagrangian, meaning that
the equilibrium solutions, for example, of the associated energy may be different from its original form without the penalty 
unless linear deformation fields are obtained.  Thus, the additional term does not vanish and its effect, unlike 
a null-Lagrangian,
 is present in general. 
A central question remains, that is, in what circumstances can one
ensure the coercivity of the original energy functionals, such as \eqref{eq:dir}, defined by the nonlocal deformation gradient.


\section{Nonlocal gradient operator and Dirichlet integral in 1D}
\label{sec:ngo}

To illustrate the key concepts,  we focus on a nonlocal gradient operator for a
scalar function field defined on a one dimensional periodic cell given by $\Omega=(0,1)$. 
 The simple setting
 allows us to more clearly present the central findings without much more
 tedious technical derivations for higher dimensional vector fields.
  It is expected that
 our approach works also for such a more general scenario.

To be more specific about the scalar nonlocal gradient operator  (or nonlocal first
derivative)  $\mbG$, we  consider suitably scaled kernels so that the local limit of  $\mbG$
as $\delta$ goes to zero recovers exactly the first order derivative $\frac{d}{d x}$ denoted by 
$\mGz$. 
A popular choice of the kernel $\rds$ is based on a rescaling $\rho_\delta(s)=\delta^{-1}\rho(s/\delta)$
with the following assumptions on $\rho$:
\beq\label{ker1}
\rho(-s)=\rho(s),\quad 
\left\{\begin{array}{ll}\rho(s)\geq 0,\;\; & \forall s\in (-1,1),\\[.1cm]
\rho(s)=0,\;\; & \forall s\notin (-1,1),
\end{array}\right.\quad\text{and}\quad
 \int_{-1}^{1} \rho(s) ds=1\,,
\eeq
with further assumptions on $\rho=\rho(s)$ to be discussed later.


Specializing \eqref{eq:ndg} for a one dimensional scalar field $u=u(x)$ that are assumed
to be square integrable and periodic with a periodic cell $\Od$, 
 we focus on the nonlocal gradient operators defined below:
\begin{eqnarray}\label{ngo}
\mbG u(x)  &=& \int_{-\delta}^\delta  \rho_\delta(|s|) \frac{u(x+s)-u(x)}{s}ds\\
&=& \int_{-\delta}^\delta  \rho_\delta(|s|) \frac{u(x)-u(x-s)}{s}ds\\
&=&\int_{-\delta}^\delta  \rho_\delta(|s|) \frac{u(x+s)-u(x-s)}{2s}ds\\
&=&\int_{0}^\delta  \rds \frac{u(x+s)-u(x-s)}{s}ds\,. \label{ngo_new}
\end{eqnarray}

Since $\rho_\delta=\rds$ can be seen as a density function defined
on $(-\delta,\delta)$, $\mbG$  is effectively a continuum weighted
 average of some discrete
first order difference operators up to the scale $\delta$.
If  $\rho_\delta =\rds$ gets localized  as $\delta\to 0$ and behaves like a 
Dirac-delta measure at the origin in the limit,  
we  may indeed see $\mGz=\frac{d}{dx}$ as the formal local limit of
$\mbG$. 

Detailed studies of operators defined by \eqref{eq:ndg} and the specialized form \eqref{ngo}
ad well as their local limits  are the subject of recently developed 
nonlocal vector calculus, see \cite{sys1} for formal derivations 
and \cite{MeDu15} and \cite{dyz16} for more extended functional analysis.  
Nonlocal analog
of integration by parts formula has also been rigorously derived \cite{dyz16,MeDu15}.

The one dimensional version of the Dirichlet integral associated with $\mbG$ and its local form can
be written as
\begin{equation}
\label{eq:do}
E_\delta(u)=\int_\Od  |\mbG u (x)|^2 dx\,\qquad
E_0(u)=\int_\Od  |\nabla u (x)|^2 dx\,,
\end{equation}
for any scalar periodic function $u$ with the periodic cell $\Od$.

Due to the periodic boundary condition, a constraint is
needed to determine the constant shift in the deformation field.
We thus only consider those functions that satisfy
\begin{equation}\label{nlceq1}
  \int_\Od u(x) dx =0 \,.
\end{equation}
 
There are also the one sided versions (see \cite{dtty16,dyz16} for related discussions)
\begin{equation}\label{ngopm}
\mathbb{G}_\delta^\pm u(x)=\pm 2 \int_0^\delta \rds \frac{u(x\pm s)-u(x)}{s}ds\,.
\end{equation}
Similarly, nonlocal diffusion operators that are analogs of the classical diffusion
or second derivative operators have also been a subject of extensive study \cite{sys1,sys2,MeDu14e,MeDu14r,MeDu15a}.
The connections between these operators are to be further discussed in the next section.

\section{Equivalent formulation of Dirichlet integral}
\label{sec:equ}

In the one dimensional case, it is particularly convenient to connect the nonlocal Dirichlet integral given by $E_\delta(u)$
in \eqref{eq:do} with another popular nonlocal version of $E_0(u)$:
\begin{equation}
\label{eq:dod}
\hat{E}_\delta(u)=\int_\Od \int \omd(s) \left|\frac{u (x+s)-u(x)}{s}\right|^2 ds dx\,.
\end{equation}
In the context of peridynamics, $\hat{E}_\delta(u)$ may be seen as the linearized bond-based elastic energy
associated with a nonlocal interaction kernel $\omd=\omd(s)$.

We now discuss the relations between \eqref{eq:do} and  \eqref{eq:dod}.
Let $D_{\delta}$ denote $(-\delta,\delta)^2$ and $H_\delta=\Od\times D_{\delta}$. 
Note first that in the principal value sense of the integrals,  we have
$$
\mbG u(x)  = \int_{-\delta}^\delta  \rho_\delta(|s|) \frac{u(x+s)}{s}ds
= -\int_{-\delta}^\delta  \rho_\delta(|s|) \frac{u(x-s)}{s}ds
$$
Thus,
\begin{eqnarray*}
E_\delta(u)&=&\int_\Od  |\mbG u (x)|^2 dx\\
   &= & - \int_{H_\delta}
 \frac{ \rho_\delta(|s|)   \rho_\delta(|t|)}{st}
  u(x+s)u(x-t)  ds dt dx\\
     &= & - \int_{H_\delta}
 \frac{ \rho_\delta(|s|)   \rho_\delta(|t|)}{st}
  u(y+s+t)u(y)  ds dt dy
   \end{eqnarray*}
  where we have done a shift in the variable $y=x-t$ but
  the domain of integration remains the same due to
    the periodicity of $u=u(x)$ in $x$.
From the above, we then notice that, after switching $y$ back to $x$, 
  \begin{eqnarray*}
E_\delta(u)       &= & - \int_{H_\delta}
 \frac{ \rho_\delta(|s|)   \rho_\delta(|t|)}{st}
  u(x+s+t)u(x)  ds dt dx\\
 &= &\int_{H_\delta}
 \frac{ \rho_\delta(|s|)   \rho_\delta(|t|)}{2st}
 [u(x+s+t)^2 -2 u(x+s+t)u(x)+  u(x)^2] ds dt dx\\
   &=&  \int_{H_\delta}
 \frac{ \rho_\delta(|s|)   \rho_\delta(|t|)}{2st}
  \left[u(x+s+t) - u(x)\right]^2
   ds dt dx\,.
   \end{eqnarray*}
Now, consider the transformation $a=s+t$ and $b=t-s$ wit $x$ unchanged,
     we use $\hat{H}_\delta$ to denote the region in the
    new variables obtained
    from the transformation of $H_\delta$, then
  \begin{eqnarray*}
E_\delta(u)      &=&
        \int_{\hat{H}_\delta}
  \rho_\delta\left(\left|\frac{a-b}{2}\right|\right)   \rho_\delta\left(\left|\frac{a+b}{2}\right|\right)
 \frac{2a^2}{a^2-b^2}
  \left| \frac{u(x+a)-u(x)}{a}\right|^2
   da db dx\\
&=&    \int_\Od \int_{-2\delta}^{2\delta} \omd(|a|)
  \left| \frac{u(x+a)-u(x)}{a}\right|^2
   da dx\,,
     \end{eqnarray*}
     where 
$$
\omd(a)=\omd(|a|)=
\int_{|a|-2\delta}^{-|a|+2\delta}
  \rho_\delta\left(\left|\frac{a-b}{2}\right|\right)   \rho_\delta\left(\left|\frac{a+b}{2}\right|\right)
 \frac{2a^2}{a^2-b^2}db 
 $$
is a kernel function supported in $a\in (-2\delta,2\delta)$. 

  This implies that $E_\delta(u)$ is equivalent to  $\hat{E}_\delta(u)$ with the kernel
$\omd=\omd(|a|)$. Consequently, the corresponding nonlocal diffusion operator (that is, the variation of
the energy,  or the bond force operator in peridynamics) is given by
$$- \mathcal{L}_\delta u(x)=\int_{-2\delta}^{2\delta} \omd(s) \frac{u(x+s)-2u(x)+u(x-s)}{s^2} ds\,.
$$
As we have elucidated in our earlier works, the nonlocal operator $\mathcal{L}_\delta$ may be
viewed as a nonlocal continuum weighted average of the classical second order central difference 
operator with the kernel $\omd$ serving as the weight function.
  A direct calculation shows that
$$\int_{-2\delta}^{2\delta} \omd(a)da = 1,$$
which gives the correct normalization condition on $\omd$. 
However, it can also be seen that
$$\int_{-2\delta}^{2\delta} \frac{\omd(a)}{a^2}da=0,$$
which implies that the nonlocal interaction kernel is a sign-changing one. That is, for example,
in the context of linear bond-based peridynamics,  we get
 both repulsive and attractive bond forces. 
 Naturally, repulsive interaction with a positive sign of kernel more likely yields coercivity and stability.
 Having attractive  interactions may cause the loss of coercivity but this is not always the case.
 In \cite{MeDu13}, well-posedness of linear
 bond-based peridynamics with a sign changing kernel has been established. In essence, as long as
 the repulsive effects are dominant, we could still expect a well-defined nonlocal model. 
 
 We take the moment to consider a couple of properties of $\omd$ in connection with
 $\rho_\delta$.
 First, we note that if $\rho_\delta$ is taken to be rescaled from a horizon ($\delta$) 
 independent
 kernel $\rho(s)$, that is
 $\rho_\delta(s)=\delta^{-1}\rho(s/\delta)$ with $\rho$ satisfying \eqref{ker1}. Then,
by a change of variables $a=\tilde{a}\delta$ and $b=\tilde{b}\delta$, we have 
 \begin{equation}
 \label{omega}
\omd(a)=\frac{1}{\delta}
\int_{|\tilde{a}|-2}^{-|\tilde{a}|+2}
    \rho(|\frac{\tilde{a}-\tilde{b}}{2}|)   \rho(\left|\frac{\tilde{a}+\tilde{b}}{2}\right|)
 \frac{2\tilde{a}^2}{\tilde{a}^2-\tilde{b}^2}d\tilde{b} =\frac{1}{\delta}\omega_1(\frac{a}{\delta})\,.
 \end{equation}
 This means, not surprisingly, that $\omd$ is also rescaled from a kernel $\omega_1$ symmetrically defined on $(-1,1)$. 
Next, we 
 note that it is easy to see
 $$\int_{-2\delta}^{2\delta} \frac{|\omd(a)|}{a^2}da\leq 
 \left( \int_{-\delta}^{\delta} \frac{\rds}{|s|}ds \right)^2\,,
 $$
 which tells us in particular that the singular behavior of $\rho_\delta$ at the origin
 likely controls the singularity of $\omd$ near the origin. Clearly, if $|s|^{-1}\rds$ is integrable,
 then so is $a^{-2}|\omd(a)|$. This in turn implies that  $\mbG$ and the associated nonlocal
 diffusion operator $\mathcal{L}_\delta$ are bounded operators on the space of square integrable functions.
 While kernels with integrable  $|s|^{-1}\rds$ and  $a^{-2}|\omd(a)|$ are among the popular choices in
 applications, they are not necessarily good choices, as shown later,  for correspondence formulations
 like \eqref{eq:do} are to be adopted.

 In the following, we provide more detailed calculations to give some explicit
 conditions on $\rho_\delta$ under which the coercivity of \eqref{eq:do} can be assured. The 
 periodic setting allows us to use Fourier
 analysis, a very convenient and frequently used tool in studies of nonlocal models,
 see for example \cite{non9}. The type of calculations involved is similar to dispersion
 analysis, see for instance \cite{bessa16,silling1,WA05} for studies related to peridynamics.
 

\section{Stability of nonlocal Dirichlet integral}
\label{sec:sta}

Let us first clarify the stability or coercivity that we refer to, namely, we define
a function space $V_\delta$ that is the completion of $C^\infty$ periodic functions
with mean zero subject to the nonlocal norm
$$\|u\|_\delta=(E_\delta(u)+\int_\Od |u(x)|^2 dx)^{1/2}.$$
The coercivity (or variational stability) of the Dirichlet integral $E_\delta(u)$ refers to the
fact that over the space $V_\delta$, we have a positive constant $C>0$, such that
$$E_\delta(u)\geq C \|u\|_\delta^2,\quad \forall u\in V_\delta.$$
Obviously, this can be seen as a consequence of the so-called nonlocal Poincar\'{e} inequality \cite{MeDu14e,MeDu14r}:
there exists a constant $c>0$ such that 
\beq
\label{eqn:pi}
E_\delta(u) = \int_\Od  |\mbG u (x)|^2 dx \geq c \int_\Od |u(x)|^2 dx
,\quad \forall u\in V_\delta.
\eeq

We note that the argument presented in \cite[Proposition 2]{silling2} for the failure of stability of $E_\delta(u)$
was based the choice of an increment in the deformation field by a Dirac-delta point measure. Such a choice
is not feasible in the space   $V_\delta$. 
 We now attempt to specify some conditions on the kernel $\rho_\delta$ so that \eqref{eqn:pi} 
can be verified and thus leading to the stability of $E_\delta(u)$.
We take advantage of the periodicity of $u$ to adopt Fourier analysis.

Under periodic conditions and the constraint \eqref{nlceq1}, 
 we could write $u$ in terms of their Fourier series, namely, 
\[
u(x) = \sum_{k=1}^\infty \widehat{u}(k) e^{i2\pi kx} \,,
\quad
\text{ where }
\quad
\widehat{u}(k)  =\int_\Od u(x)  e^{-i 2\pi kx} dx  \;.
\]

The rest of the technical derivations is focused on one dimensional scalar fields. For
the vector field case in high dimensions, we refer to \cite{dt17}.  Let us first present
the so-called Fourier symbols of the operator $\mbG$.

\begin{lem}
\label{bform}
The Fourier transform of the nonlocal gradient operator $\mbG$ is given by
\begin{align}
\widehat{\mbG u} (k)=i b_\del(k) \widehat{u}(k)
\end{align}
where  $b_\del(k) $ are the Fourier symbols  of $\mbG$ given by
\begin{align}
b_\del(k)  = 2 \int_0^\delta \frac{\rds}{s}\sin(2\pi ks) ds
 \,. \label{coef_b}
\end{align}
\end{lem}

\begin{proof}
Using \eqref{ngo_new} and the periodicity of $u$, we have
\[
\begin{split}
\widehat{\mbG u} (k) &= \int_\Od \mbG u(x) e^{-i 2\pi kx} dx  \\
&= \int_\Od  \int_{0}^\delta  \rds \frac{u(x+s)-u(x-s)}{s} e^{-i 2\pi kx} ds dx \\
&= \int_\Od  \int_{0}^\delta  \frac{\rds}{s}(e^{i 2\pi ks}-e^{-i 2\pi ks} )  u(x) e^{-i 2\pi kx}  ds dx  \\
&= 2 i \int_0^\delta \frac{\rds}{s}\sin(2\pi ks) ds \int_\Od  u(x) e^{-i 2\pi kx} dx \\
&= i b_\del(k) \widehat{u}(k)\,.
\end{split} 
\]
\end{proof}

The following
simple fact on the sine Fourier coefficient, a special form of Riemann-Stieltjes type integrals,
is useful to our discussion.

\begin{lem} \label{lem:positivity}
Given a measurable, non-negative and non-increasing function $g=g(x)$ with $xg(x)$ integrable, we have
\beq
\label{fou1}
\int_0^{2\pi} g(x) \sin(x) dx \geq 0\,
\eeq
with the equality holds only for $g$ being a constant function. Consequently, for any 
 $h>0$ and $a>0$, we have
\beq
\int_0^h g(x) \sin(ax) dx \geq 0\,,
\eeq
with the equality holds only for $g$ being a constant function (and with value zero if $ha$ is not an integer multiple of $2\pi$).
\end{lem}
\begin{proof} The inequality \eqref{fou1} follows immediately from the observation that
\[\int_0^{2\pi} g(x) \sin(x) dx  = \int_0^{\pi} [g(x)-g(x+\pi)] \sin(x) dx \geq 0.
\]
By the non-increasing property, we see that the  equality holds only for $g$ being a constant function. 
The more general case follows by applying a change of variable and taking a zero extension of $g$ outside $(0,h)$ to
cover complete periods of the scaled sine function.
\end{proof}

The expressions of  Fourier symbols $\{b_\del(k)\}$ given in Lemma \ref{bform}
have a number of immediate consequences upon simple observations.
First, if as $\delta\to 0$ we have $2\rds$, which is a density on $(0,\delta)$, approaches to the
Dirac-delta measure at the origin, then for any given $k$ and $\delta\to 0$,
$$b_\del(k)  = 2 \int_0^\delta \frac{\rds}{s}\sin(2\pi ks) ds = 2\pi k \int_0^\delta 2 \rds \frac{\sin(2\pi ks)}{2\pi ks} ds
\to 2\pi k\, ,
$$
which again gives the obvious connection of the nonlocal derivative to the local derivative.

Another immediate consequence, coming from Lemma \ref{lem:positivity},  is
that  $b_\del(k)$ stays positive for any finite wave number $k$ 
if $s^{-1} \rds$ is non-increasing and not a constant. This provides a sufficient condition to assure $b_\del(k)\neq 0$. 
While not necessary, violating such a condition may result in undesirable effects. 
We can see from the simple example that $s^{-1}\rds$ being a constant function leads to
$b_\del(k)=0$ for some wave number $k$. Indeed, for such a choice of $\rho_\delta$,
 $$b_\del(k)= c \int_0^\delta \sin(2\pi ks) ds = \frac{c}{2\pi k}(1-\cos(2\pi k \del)),$$ 
 which is zero if $k\delta$ is an integer. The Fourier modes corresponding to such wave numbers 
 may be interpreted as the so-called {\em zero-energy} modes, reinforcing
 the understanding that there are inherent instabilities associated with the correspondence theory at the continuum
 level \cite{silling2}. However, such zero-energy modes do not exist for proper choices of the kernel  described below.
  
 Naturally, the  coercivity of the Dirichlet integral is not only concerned with $b_\del(k)$ for finite $k$ but also 
 its asymptotic and uniform behavior as $k\to\infty$. The latter is a feature of the continuum
 theory that allows the wave number going to infinity, though, unfortunately, it has not been given adequate attention
 in the existing literature
 before.
  From the Riemann-Lebesgue lemma, we know that $b_\del(k)\to0$ as $k\to\infty$ if $s^{-1}\rds$ is integrable.
  Thus, requiring $s^{-1}\rds$ being non-integrable becomes necessary, a fact that should be taken into
  consideration seriously when working with the correspondence theory.
  
 Note that requiring  $s^{-1}\rds$ both  non-increasing and non-integrable, 
 in order to get the non-degeneracy  of $b_\del(k)$ for finite $k$ and as  $k\to\infty$, 
 would lead to a
restricted growth of $\rds$ near the origin. That is, as the undeformed bond length
 approaches zero, the prescription of the nonlocal interaction is limited.  A most general
discussion about necessary and sufficient conditions on the kernel to assure the coercivity is beyond the scope of
this work. Instead, to keep the mathematical discussions to a minimal level, 
we present a result on the coercivity of the nonlocal Dirichlet integral under the following sufficient conditions.

\begin{assu} \label{assu:frackernel}
We assume that the kernel $\rho_\delta$ is given by a rescaled kernel $\rho_\delta(s)=\rho(s/\delta)/\delta$  with
 $\rho=\rho(s)$ satisfying \eqref{ker1} and the following:\\
{\rm 1)} ${\rho(s)}/{s}$ is non-increasing for $s\in (0,1)$;\\
{\rm 2)} $\rho(s)$ is of fractional type in at least a small neighborhood of origin, namely, 
 there exists some $\ep>0$ such that for $s\in (0,\ep)$ we have
\begin{equation} \label{frackernel}
\rho(s)= \frac{c}{s^{\alpha}} \,,
\end{equation}
for some constant $c>0$ and $\alpha\in(0,1)$.
\end{assu}

Note that the condition associated with \eqref{frackernel} automatically implies the non-integrability of  $s^{-1}\rds$ in
the one space dimensional case under consideration.

\begin{thm}
  \label{thm:poincare}
Under the Assumption \ref{assu:frackernel}, the nonlocal Poincar\'e inequality holds, namely, there exists a constant $C$
 independent of $u$ such that
\beq \label{eq:poincare}
\| u\|^2_{L^2(\Od)}\leq C E_\delta(u)\,.
\eeq
Moreover, $C$ is also independent of $\delta$ as $\delta\to0$.
\end{thm}

\begin{proof}
Under periodic conditions, \eqref{eq:poincare} is equivalent to showing that $b_\del(k)$ is uniformly bounded from below. 
In fact,  we can write
\[
\begin{split}
b_\del(k)= 2 \int_0^\delta \frac{\rds}{s}\sin(2\pi ks) ds = \frac{2}{\del}   \int_0^1 \frac{\rho(s)}{s}\sin( a s ) ds \\
\end{split}
\]
where $a= 2\pi k\delta$. Notice that the above quantity is positive with the assumption that $\rho(s)/s$ is non-increasing
and not a constant.

For $a<1$, we use 
$$\sin(a s) \geq as -\frac{(as)^3}{6}$$ to get
\[
b_\del(k) \geq \frac{2a}{\delta}\int_0^1 \rho(s) ds - \frac{a^3}{3\delta} \int_0^1 \rho(s) s^2 ds \geq \frac{5 a }{3 \delta}=\frac{10\pi}{3} k
\geq \frac{10\pi}{3}  \,.
\]
For $a\in [1, {2\pi}/{\ep}]$, where $\ep$ is the parameter defined in Assumption \ref{assu:frackernel}, we know that 
$$ 2 \int_0^1 \frac{\rho(s)}{s} \sin(a s) ds >0$$ 
is a continuous function of $a$,  thus it has a lower bound $ C_1$, then
\[
b_\del(k) \geq  \frac{ C_1}{\del} \geq  C_1 \ep k\,.
\]
Now for $a> {2\pi }/{\ep}$, we write
\[
b_\del(k) =\frac{2}{\del}   \int_0^1 \frac{\rho(s)}{s}\sin( a s ) ds =\frac{2}{\del}  \left( \int_0^{\frac{2\pi}{a}} \frac{\rho(s)}{s}\sin( a s ) ds + \int^1_{\frac{2\pi}{a}} \frac{\rho(s)}{s}\sin( a s ) ds\right)\,.
\]
Using Lemma \ref{lem:positivity} and the assumption that $\rho(s)/s$ is non-increasing on $(0,1)$, we have 
\[
\int^1_{\frac{2\pi}{a}} \frac{\rho(s)}{s}\sin( a s ) ds  = \int_{0}^{1-\frac{2\pi }{a}} \frac{\rho(s+{2\pi}/{a})}{s+{2\pi}/{a}}\sin( a s ) ds \geq 0 \,.
\]
 Then 
\begin{equation}
\label{eq:lower}
b_\del(k) \geq    \frac{2}{\del}   \int_0^{\frac{2\pi}{ a} } \frac{\rho(s)}{s}\sin(as ) ds \geq \frac{ c a^\alpha}{\del}   \int_0^{2\pi} \frac{1}{s^{1+\alpha}} \sin(s) ds =\frac{\tilde C k^\alpha}{\del^{1-\alpha}}   \,,
\end{equation}
where the Assumption \ref{assu:frackernel} is used with $\alpha\in(0,1)$. 
Thus $b_\del(k)$ has a lower bound for all integer $k\geq 1$ independent of $\delta$, 
 which implies the Poincar\'e inequality \eqref{eq:poincare}.
\end{proof}

Theorem \ref{thm:poincare} shows the stability of the Dirichlet integral while the kernel $\rho(s)$ a fractional-type kernel near origin with the fractional component $\alpha\in(0,1)$.
 We can further see from  its proof that  the generated function space is equivalent to a fractional Sobolev space. 
\begin{thm}
Under the Assumption \ref{assu:frackernel}, for each fixed $\delta>0$, the space $V_\delta$ defined by the nonlocal norm $\| \cdot \|_{\delta}$ is equivalent to the fractional Sobolev space $H^\alpha(\Omega)$, where $\alpha\in (0,1)$
is the index defined in \eqref{frackernel}. 
\end{thm}
\begin{proof}
In the proof of Theorem \ref{thm:poincare}, notice that in the case $a>{2\pi}/{\ep}$, we first have 
the lower bound given by \eqref{eq:lower}. Moreover, we could also write
\[
b_\del(k) =\frac{2}{\del}   \int_0^1 \frac{\rho(s)}{s}\sin( a s ) ds =\frac{2}{\del}  \left( \int_0^{\frac{\pi}{a}} \frac{\rho(s)}{s}\sin( a s ) ds + \int^1_{\frac{\pi}{a}} \frac{\rho(s)}{s}\sin( a s ) ds\right).
\]
Using Lemma \ref{lem:positivity} once again,  we have
\[
\begin{split}
\int^1_{\frac{\pi}{a}} \frac{\rho(s)}{s}\sin( a s ) ds =&  \int_{0}^{1-\frac{\pi }{a}} \frac{\rho(s+{\pi}/{a})}{s+{\pi}/{a}}\sin( a s +\pi) ds \\
=& -\int_{0}^{1-\frac{\pi }{a}} \frac{\rho(s+{\pi}/{a})}{s+{\pi}/{a}}\sin( a s) ds <0 \,.
\end{split}
\]
Thus we obtain the upper bound given by
\[
\begin{split}
&b_\del(k) \leq  \frac{2}{\del}   \int_0^{\frac{\pi}{ a} } \frac{\rho(s)}{s}\sin(as ) ds \leq \frac{ c a^\alpha}{\del}   \int_0^{\pi} \frac{1}{s^{1+\alpha}} \sin(s) ds =\frac{\tilde C k^\alpha}{\del^{1-\alpha}}   \,.
\end{split}
\]
Combining both the lower and upper bounds, we get the equivalence to the fractional space and norm.
\end{proof}

We remark that an alternative way to see the above is to use the equivalent formulation given in \eqref{eq:dod} for the
kernel $\omega_\delta$ in \eqref{omega}. For the kernel $\rho(s)$ in \eqref{frackernel} that grows like $s^{-\alpha}$ near the 
origin, a direct calculation shows that $\omega_\delta(s)$ grows like $s^{1-2\alpha}$. Hence,
the nonlocal Dirichlet integral gives a canonical form of the square of a $H^{\alpha}$ semi-norm.

The equivalence of function spaces mentioned above is not simply a mathematical statement, it 
too bears significance in nonlocal modeling of physical
processes involving singularities such as the peridynamic modeling of cracks. 
In practice, in order to allow discontinuous solutions in the underlying energy space $V_\delta$, we see that one should make
$\alpha<1/2$ (in one space dimension) by the standard Sobolev space embedding result. 
Based on the above theorems and the characterization on the kernels given in the Assumption  \ref{assu:frackernel}, we see
that there are indeed reasonable choices of the nonlocal interaction kernels that provide coercive (and stable) forms
of energy for the correspondence theory (to maintain a well-behaved mathematical model), while allowing the discontinuous deformation field (to keep a physically desirable feature).
 However, these possible choices may be limited, for example, they may be subject to conditions given in Assumption  \ref{assu:frackernel}. Note again, we do not claim that the assumption here is the most general one possible since
 our objective is to establish some rigorous results with fairly elementary calculations without making
 the mathematical derivations too technical. 

\section{Other nonlocal variants of the Dirichlet integrals}
\label{sec:oth}

In \cite{silling2}, an alternative formulation to the elastic energy is provided as a possible remedy to alleviate
the loss of coercivity of the original correspondence peridynamic materials models. The main ingredient
consists of an additional contribution involving the nonuniform part of the deformation field $u$ that is denoted by $z$.
In the linearized theory and one space dimension,  we have
$$z(y,x)=u(y+x)-u(y) - \mbG u (y) x.$$
Note that if taking $z=z(y,x)$ as a peridynamic state \cite{silling3} (with dependence on both $y$ and $x$),
we have
$$\mbG z(y,x)=\mbG u (y+x) - \mbG u (y), $$
which is  in general nonzero for nonlinear deformation field $u$.
The {\em stabilized} energy  suggested in \cite{silling2} is given by the following variant of the 
nonlocal Dirichlet integral:
\begin{equation}
\label{eq:pe}
\tilde{E}_\delta(u)=\int_\Od  |\mbG u (x)|^2 dx+
 \int_\Od \int_{-\delta}^\delta  \sigma_\delta(|s|) \left|\frac{u (x+s)-u(x)}{s} - \mbG u (x)
\right|^2 ds dx\,,
\end{equation}
where $\sigma_\delta=\sigma_\delta(|s|)$ is assumed to be a compactly supported kernel for $|s|\leq \delta$. In this case,
the coercivity of $\tilde{E}_\delta$ can be established under some conditions on $\sigma_\delta$ but
without imposing the stronger conditions
on the kernel $\rho_\delta=\rds$ given in the Assumption \ref{assu:frackernel}.
Here, we present an argument that is different from that given in \cite{silling2}. Let us denote
$$ \int_{-\delta}^\delta  \sigma_\delta(|s|)ds=\beta>0\,.$$
Then
\begin{eqnarray*}
\tilde{E}_\delta(u)&=&\int_\Od  |\mbG u (x)|^2 dx+
 \int_\Od \int_{-\delta}^\delta  \sigma_\delta(|s|) \left|\frac{u (x+s)-u(x)}{s} - \mbG u (x)
\right|^2 ds dx\\
&=&\int_\Od  (1+\beta) |\mbG u (x)|^2 dx+
 \int_\Od \int_{-\delta}^\delta  \sigma_\delta(|s|) \left|\frac{u (x+s)-u(x)}{s}\right|^2 ds dx\\
 &&\qquad
 -2 \int_\Od \int_{-\delta}^\delta  \sigma_\delta(|s|) \frac{u (x+s)-u(x)}{s} \mbG u(x)  ds dx.
 \end{eqnarray*}
 Notice that 
 $$ 2 \left| \frac{u (x+s)-u(x)}{s} \mbG u(x)\right| \leq \frac{2}{2+\beta} 
 \left| \frac{u (x+s)-u(x)}{s} \right|^2 + \frac{2+\beta}{2} |\mbG u(x)|^2\,.
 $$ 
 So 
 \begin{eqnarray*}
\tilde{E}_\delta(u)&\geq & \frac{\beta}{2}\int_\Od  |\mbG u (x)|^2 dx
+ \frac{\beta}{2+\beta} 
 \int_\Od \int_{-\delta}^\delta  \sigma_\delta(|s|) \left|\frac{u (x+s)-u(x)}{s}\right|^2 ds dx\,.
\end{eqnarray*}
Now we can see that the second term represents a typical energy for a linear bond-based model and,
for a positive $\sigma_\delta=\sigma_\delta(s)$, the term
vanishes only for a constant field that is identically zero by \eqref{nlceq1}.

In fact, 
under suitable assumptions on $\sigma_\delta$ as those presented in \cite{BBM} and \cite{Ponce}
(and extended to vector fields in \cite{MeDu13,MeDu14e,MeDu14r}), we have the Poincar\'{e}
inequality 
$$
 \int_\Od \int_{-\delta}^\delta  \sigma_\delta(|s|) \left|\frac{u (x+s)-u(x)}{s}\right|^2 ds dx \geq c\int_\Od |u(x)|^2dx
 $$
 for a constant $c>0$, hence the coercivity of the energy  $\tilde{E}_\delta$.  This line of argument
 has a striking similarity with the study of coercivity of the  energy functionals for
 the linearized state-based peridynamic Navier
 equation. The state-based energy functional involves contributions from two terms similar to
 those in \eqref{eq:pe} but with different mechanical interpretations: 
 one from the elongational part and the other
 from the deviatoric part, see \cite{MeDu14e} for more details.

While the coercivity of $\tilde{E}_\delta$ can be guaranteed, such a variant presents other issues. 
For instance, one issue with \eqref{eq:pe} is that while for linear deformation field the energy remains independent of $\beta$,
there are obvious discrepancies when more general nonlinear deformation fields are considered, leading to
different physical responses that are dependent on the choices of $\beta$. The same can be said about
the variational problems related to \eqref{eq:pe} subject to the work of an external force.
The equilibrium solutions, for the same given external force, will be generically different when $\beta$ changes.
Similarly, the Fourier
symbols associated with the modified operator (and thus the dispersion relations) also differ from their
original and un-modified forms.
Thus, the extra penalty term could play a significant role for nonzero $\beta$, which may become an undesirable
feature in practice. 
In contrast, for more careful choices of the kernel $\rho_\delta$, our analysis in the
previous section shows that there is no need to introduce the penalty term for the sake of coercivity or
variational stability. 

Besides the variant discussed above, there are also other options to get a nonlocal Dirichlet integral that is different
from \eqref{eq:do}. For example, we may replace the nonlocal gradient operator in \eqref{ngo} by
the one-sided versions given in \eqref{ngopm}. Similar discussions can be made with these new choices,
along with extensions to vector fields defined in higher space dimensions
(see \cite{dt17} for related calculations).



\section{Implications and generalizations}
\label{sec:con}

An important message taken from the investigation  presented in this work is that one should evaluate carefully the choices 
of the nonlocal interaction kernels when correspondence models are adopted. An intuitive interpretation of our rigorous
analysis is that strengthened
interactions among close-by materials points tend to promote stability. This is perhaps not surprising as the well-posed local
model, when physically valid, represents the extreme case that all interactions are concentrated at the same materials points.
On the other hand, in order to allow solutions with defects and singularities like discontinuities in the deformation field for periydnamics, it is equally important to have the interactions spread over a nonlocal region. Moreover, the 
close-by interactions should not be too strong to disallow the formation of such singular behavior. Hence, adopting
appropriate nonlocal interaction kernels becomes a subtle issue when one desires to take advantage of the correspondence theory to model complex systems.
This finding sheds light on the range of
applicability of peridynamic correspondence material models. The latter should not be adopted blindly but needs not
be thrown out entirely. Via Fourier analysis, we are able to offer
a rather precise characterization on the feasible choices of nonlocal interaction kernels that helps
maintaining coercivity and stability on the continuum level. 

We note that our discussion here is
focused on the continuum level since nonlocal models like peridynamics are indeed fundamentally
continuum theory in the first place.  The stability and coercivity issues should  thus be
variational (or continuum model) properties, independent of discretizations. While the discussion
is focused on one dimensional scalar fields, generalizations to high dimensional vector fields
are possible, and one can find relevant discussions in \cite{dt17}. One may also
consider geometric settings other than the periodic domain, in which case, Fourier analysis is no
longer effective, but  one may use techniques similar to those in \cite{MeDu13} for peridynamic
 models with a sign-changing kernel to study the stability of energy with a more general class of kernels
subject to various nonlocal constraints.

There are inevitably additional issues when numerical discretizations are concerned. 
In fact, a side-effect of the strengthened
interaction for close-by materials points is that more accurate quadrature rules may have to be adopted in discretization.
Still,  one may follow
similar ideas presented here to reach further understanding of the numerical stability and coercivity issues as well.
 This would allow us to delineate
the roles of physical scale and level of numerical resolutions, a point that is worthwhile to be emphasized
 for nonlocal modeling.
 One can also attempt to develop asymptotically
compatible discretizations \cite{dtty16,non10,non11} to the correspondence models so as to retain 
consistency and robustness of 
the numerical approximations.
In addition, the current study can further enable us to introduce mixed formulation to numerically
approximate the nonlocal models
 based on the correspondence material models whenever the latter are physically sound and mathematically
 valid.  Lastly, 
the notion of nonlocal gradient may
also be related to the use of kernel-based integral approximations to differential operators in
methods like SPH and RKPM 
 \cite{dt17,WK:96,mono:05}. Future studies can help build a stronger connections between these similar subjects.

\noindent{\bf Acknowledgement}: The authors would like to thank S. Silling, M. Gunzburger, R. Lehoucq, J. Foster, F. Bobaru,
J.S. Chen, W.K. Liu, and in particular T. Mengesha,  for discussions on related subjects.

\bibliographystyle{siam}

\end{document}